\documentclass[letterpaper,conference]{IEEEtran}

\usepackage{graphicx}
  \usepackage{epstopdf}
\usepackage{subfig}
\usepackage{amsmath}
\usepackage{amssymb}
\usepackage{pgf}
\usepackage{tikz}
\usepackage{enumerate}
\usepackage{url}
\usepackage{caption}
\usetikzlibrary{backgrounds,automata}

\newtheorem{theorem}{Theorem}[section]

\newtheorem{proposition}[theorem]{Proposition}

\newtheorem{problem}[theorem]{Problem}

\newtheorem{rem}[theorem]{Remark}
\newtheorem{ex}[theorem]{Example}

\newenvironment{proof}[1][Proof]{\begin{trivlist}
\item[\hskip \labelsep {\bfseries #1}]}{\end{trivlist}}

\IEEEoverridecommandlockouts

\begin{document}
\title{Memoryless Control Design for \\ Persistent Surveillance under Safety Constraints }

\author{\IEEEauthorblockN{Eduardo Arvelo, Eric Kim and Nuno C. Martins}
\IEEEauthorblockA{Department of Electrical and Computer Engineering and 
Maryland Robotics Center\\
Univeristy of Maryland,College Park, MD 20742\\
Email: \{earvelo, eskim727, nmartins\}@umd.edu } 

}

\maketitle

\begin{abstract}
%\boldmath
This paper deals with the design of time-invariant memoryless control policies for robots that move in a finite two-dimensional lattice and are tasked with persistent surveillance of an area in which there are forbidden regions. We model each robot as a controlled Markov chain whose state comprises its position in the lattice and the direction of motion. The goal is to find the minimum number of robots and an associated time-invariant memoryless control policy that guarantees that the largest number of states are persistently surveilled without ever visiting a forbidden state. We propose a design method that relies on a finitely parametrized convex program inspired by entropy maximization principles. Numerical examples are provided.
\end{abstract}

\IEEEpeerreviewmaketitle

\section{Introduction}

%p1problem statement
We develop a method to design memoryless controllers for robots that move in a finite two-dimensional lattice. The goal is to achieve persistent surveillance. The term ``persistent surveillance" is used to denote the task of \emph{continuously} visiting the largest possible set of points in the lattice. In our setup, we also impose safety constraints that dictate that certain regions are \emph{forbidden}. The forbidden regions may represent areas in which robots cannot operate (such as bodies of water) or are not allowed to visit (such as restricted airspace). The goal is to deploy the minimum number of robots equipped with a  control policy that guarantees persistent surveillance of the \emph{largest} possible set of lattice points without ever visiting a forbidden region. The memoryless strategies proposed here are applicable to miniature robots that have severe computational constraints.

\vspace{1mm}
The concept of persistent survaillance is similar to the concept of coverage \cite{Choset:2001uo}, but differs from it in that the area to be surveilled must be revisited infinitely many times. Control design for persistent surveillance has been studied in \cite{Nigam:2012hv, Nigam:2008hh}, where a semi-heuristic control policy that minimizes the time between visitations to the same region is proposed, and in \cite{Hokayem:2007cx}, which proposes an algorithm for persistent surveillance of a convex polygon in the plane. These approaches, however, are not restricted to memoryless policies and do not consider safety constraints.  On the implementation front, system architectures for unmanned aerial vehicles have been designed specifically for persistent surveillance purposes \cite{Bethke:2010ba, Michael:2011bx}.

\vspace{1mm}
%p3: Approach and Contribution
In this paper, we model each robot as a fully-observed controlled Markov chain with finite state and control spaces. This approach, which has been successfully used in the context of navigation and path planning (such as in \cite{Simmons:1995td,Undurti:2011dn,Theocharous:2002cx,Burlet:2004kz}), allows for the development of robust and highly scalable algorithms. Without loss of generality, we consider robots whose state is taken as its position on a finite two-dimensional lattice and direction of motion (taken from a set of four possible orientations), and limit the control space to two control actions (``forward" and ``turn right"). The limitation in the control space illustrates how constrained actuation can be incorporated in our formulation. It is important to highlight, however, that the ideas described in this paper can be extended to more general dynamics and state/control spaces. 

\vspace{1mm}
We use a recent result in \cite{earvelo} to compute the largest set of states that can be persistently surveilled under safety constraints, and an associated memoryless control policy. The proposed solution relies on a finitely parametrized convex program, which is highly scalable and can be efficiently solved by standard convex optimization tools, such as \cite{cvx}. The approach is based on the fact that the probability mass function that maximizes the entropy under convex constraints has maximal support \cite{Csiszar:2004vd}. We also show that the minimum number of robots needed to perform persistent surveillance of the largest set of states (without ever violating the safety constraint) is equal to the number of recurrent classes of the closed loop Markov chain under the control policy computed by the proposed convex program. The recurrent classes can be found by traversing the graph of the closed loop Markov chain.

\vspace{1mm}
%\subsection{Paper Organization}
The remainder of this paper is organized as follows. Section \ref{sec:prelim} provides notation, basic definitions and the problem statement. The convex program that computes the maximal set of persistently surveilled states and its associated control policy is presented in Section \ref{sec:recurrent}. Section \ref{sec:deploy} provides details on computing the smallest deployment of robots necessary for maximal persistently surveillance. We discuss limiting behavior and use of additional constraints in Section \ref{sec:limiting}. Conclusions are provided in Section~\ref{sec:conc}. Numerical examples are given throughout the paper to illustrate concepts and the proposed methodology.

\section{Preliminaries and Problem Statements}\label{sec:prelim}

The following notation is used throughout the paper: 
\vspace{.2 cm}

\begin{tabular}{l l}
	$\mathbb{X} \times  \mathbb{Y} $ & set of lattice positions\\
	$\mathbb{O}$ & set of orientations\\
	$\mathbb{S}\overset{def}{=}\mathbb{X} \times  \mathbb{Y}\times \mathbb{O}$ & set of robot states\\
	$\mathbb{F} \subset \mathbb{S}$ &set of forbidden states\\
	$\mathbb{U}$ & set of control actions \\ 
	%$(X_k~Y_k)$ & location of the robot at time $k$\\
	%$\Theta_k$ & orientation of the robot at time $k$\\
	$S_k$ & state of the robot at time $k$\\
	$U_k$ & control action at time $k$\\
	%$\mathbb{P}_\mathbb{X}$ & set of all pmfs with support in $\mathbb{X}$\\
	%$\mathbb{P}_\mathbb{U}$ &  set of all pmfs with support in $\mathbb{U}$ \\
	%$\mathbb{P}_\mathbb{XU}$ &  set of all joint pmfs with support in $\mathbb{X} \times \mathbb{U}$\\
	%$\mathbb{S}_{f}$ & support of a pmf $f$
\end{tabular}
\vspace{.2cm}

The state of the robot will be graphically represented as shown in Fig. \ref{fig:state}.
\begin{figure}[!ht]
  \captionsetup[subfigure]{labelformat=empty}
  \centering
  \subfloat[][${~~~~~~~(1,2,U)}$]{
	\begin{tikzpicture}[inner sep=1mm, xscale = 0.8, yscale = 0.8]
	%\draw [->] (0,2) -- (3.2,2);
	%\draw [->] (0,2) -- (0,-.2);
	\draw[step=1cm, very thin] (0,0) grid (3,2);
	\draw[fill=black!20] (0.5,0.5)--(0,1)--(1,1)--cycle;
	%\node at (1,1.5) [circle,draw=black!50,fill=black!20] {};
	\draw[thin] [->] (0.5, 0.5) -- (0.5,0.95);
	
	\node at (0.5, 2.23) [draw=none,fill=none] {$1$};
	\node at (1.5, 2.23) [draw=none,fill=none] {$2$};
	%\node at (1.5, 2.23) [draw=none,fill=none] {$3$};
	\node at (2.5, 2.23) [draw=none,fill=none] {$3$};
	%\node at (3.5, 2.23) [draw=none,fill=none] {$4$};
	%\node at (-0.2, 0.5) [draw=none,fill=none] {$3$};
	\node at (-0.2, 0.5) [draw=none,fill=none] {$2$};
	\node at (-0.2, 1.5) [draw=none,fill=none] {$1$};
	%\node at (3.4, 2) [draw=none,fill=none] {$x$};
	%\node at (0, -0.4) [draw=none,fill=none] {$y$};
\end{tikzpicture}}\quad
  \subfloat[][${(3,1,R)}$]{
	\begin{tikzpicture}[inner sep=2mm, xscale = 0.8, yscale = 0.8]
	%\draw [->] (0,2) -- (3.2,2);
	%\draw [->] (0,2) -- (0,-.2);
	\draw[step=1cm, very thin] (0,0) grid (3,2);
	\draw[fill=black!20] (2.5,1.5)--(3,1)--(3,2)--cycle;
	%\node at (1,1.5) [circle,draw=black!50,fill=black!20] {};
	\draw[thin] [->] (2.5, 1.5) -- (2.95,1.5);	
	\node at (0.5, 2.23) [draw=none,fill=none] {$1$};
	\node at (1.5, 2.23) [draw=none,fill=none] {$2$};
	%\node at (1.5, 2.23) [draw=none,fill=none] {$3$};
	\node at (2.5, 2.23) [draw=none,fill=none] {$3$};
	%\node at (3.5, 2.23) [draw=none,fill=none] {$4$};
	%\node at (-0.2, 0.5) [draw=none,fill=none] {$3$};
	%\node at (-0.2, 0.5) [draw=none,fill=none] {$2$};
	%\node at (-0.2, 1.5) [draw=none,fill=none] {$1$};
	%\node at (3.4, 2) [draw=none,fill=none] {$x$};
	%\node at (0, -0.4) [draw=none,fill=none] {$y$};
\end{tikzpicture}}
  \caption{Graphical representation of the state of the robot. In this examples, we use $\mathbb{X} = \{1, 2, 3\}$, $\mathbb{Y} = \{1, 2\}$, and $\mathbb{O} = \{R, U, L, D\}$, where $R, U, L \text{ and }D$ represent right, up, left and down directions, respectively. }
  \label{fig:state}
\end{figure}
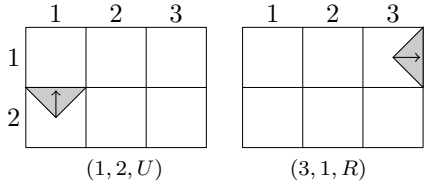

%Recursion
The robot's dynamics are governed by the (conditional) probability of $S_{k+1}$ given the current state $S_k$ and control action $U_k$, and are denoted as:
\begin{align*}
	\mathcal{Q}(s^+&, s, u)\overset{def}{=}P(S_{k+1} = s^+~\big|~ S_k = s, U_k = u),
\end{align*}
where $s , s^+ \in \mathbb{S}, u \in \mathbb{U}$.

%Control Policy
We denote any memoryless control policy by
\begin{align*}
	\mathcal{K}(u, s) \overset{def}{=}P(U_k = u~\big|~S_k = s), \qquad u \in \mathbb{U}, s \in \mathbb{S},
\end{align*}
where $\sum_{u \in\mathbb{U}}\mathcal{K}(u,s)=1$ for all $s$ in $\mathbb{S}$. The set of all such policies is denoted as $\mathbb{K}$. Note that the computation of a control action may be deterministic (when $\mathcal{K}(u,s)=1$ for a given action $u$) or carried out in a randomized manner, in which case the policy dictates the probabilities assigned to each control action for a given state.
\vspace{1mm}
\paragraph*{Assumptions}\
\begin{itemize}
\item Throughout the paper we assume that $\mathcal{Q}$ is given. Hence, all quantities and sets that depend on the closed loop behavior may be indexed only by the underlying control policy $\mathcal{K}$.
\item When multiple robots are considered, we assume that they are identical and have dynamics governed by  $\mathcal{Q}$. In these situations, every robot executes the same control policy. Moreover, multiple robots are allowed to occupy the same position. 
\end{itemize}

Given a control policy $\mathcal{K}$, the conditional state transition probability of the closed loop is represented as:
\begin{align*}
	\mathcal{P}_{\mathcal{K}}(S_{k+1}=s^+ \big| S_k=s) \overset{def}{=}  \sum_{u\in\mathbb{U}}\mathcal{Q}(s^+,s,u)\mathcal{K}(u,s). %\qquad s^+,s \in \mathbb{S}.
\end{align*}

We will also refer to this quantity as $\mathcal{Q}_{\mathcal{K}}(s^+, s) \overset{def}{=} \mathcal{P}_{\mathcal{K}}(S_{k+1}=s^+ \big| S_k=s)$.

%RECURRENCE AND PS
\subsection{Recurrence and Persistent Surveillance}

A state $s \in \mathbb{S}$ is \emph{recurrent} under a control policy $\mathcal{K}$ if the probability of a robot revisiting state $s$ is one, that is:
\begin{align}\label{eqn:pr}
P_{\mathcal{K}} (S_{k} = s \text{ for some } k>0~\big|~S_0 = s)=1.
\end{align} 

We define the \emph{set of recurrent states} $\mathbb{S}^{R}_{\mathcal{K}}$ under control policy $\mathcal{K}$ as follows:
\begin{align*}
	&\mathbb{S}^{R}_{\mathcal{K}} \overset{def}{=}
	\Big\{s \in \mathbb{S}: (\ref{eqn:pr})~ holds \Big\}.
\end{align*}

\begin{rem}\label{rem1} Membership in $\mathbb{S}^{R}_{\mathcal{K}}$ guarantees that once a state is visited, it will be revisited infinitely many times under control policy $\mathcal{K}$. It does not, however, guarantee that each state in $\mathbb{S}^{R}_{\mathcal{K}}$ will be visited for all initial states in $\mathbb{S}^{R}_{\mathcal{K}}$ because $\mathbb{S}^{R}_{\mathcal{K}}$ may contain multiple recurrent classes. In fact, a robot will visit a certain recurrent state $s$ with probability one if and only if it is initialized in the same recurrent class. Moreover, note that once a robot enters a recurrent class, it will never exit under control policy $\mathcal{K}$.\\
\end{rem}

We say a state $s$ is \emph{persistently surveilled} under control policy $\mathcal{K}$ and initial state $s_0 \in \mathbb{S}$ if it is recurrent under $\mathcal{K}$ and 
\begin{align}\label{eqn:ps}
P_{\mathcal{K}} (S_{k} = s \text{ for some } k>0~\big|~S_0 = s_0)=1.
\end{align} 

If a state $s$ is persistently surveilled under control policy $\mathcal{K}$ and initial state $s_0 \in \mathbb{S}^{R}_{\mathcal{K}}$, then it must be that $s$ and $s_0$ belong to the same recurrent class. 
 
We define the \emph{set of persistently surveilled states $\mathbb{S}^{ps}_{s_0,\mathcal{K}}$} under control policy $\mathcal{K}$ and initial state $s_0 \in \mathbb{S}$ to be:
\begin{align*}
	&\mathbb{S}^{ps}_{s_0,\mathcal{K}}\overset{def}{=}\
	\Big\{s \in \mathbb{S}^{R}_{\mathcal{K}}: (\ref{eqn:ps})~ holds \Big\}.
\end{align*}

The set $\mathbb{S}^{ps}_{s_0,\mathcal{K}}$ is a recurrent class of the closed loop dynamics $\mathcal{Q}_\mathcal{K}$. Note that for every state $s$ in $\mathbb{S}^{ps}_{s_0,\mathcal{K}}$, it holds that $\mathbb{S}^{ps}_{s,\mathcal{K}} = \mathbb{S}^{ps}_{s_0,\mathcal{K}}$. Moreover, if there exists a recurrent state for which $\mathbb{S}^{ps}_{s_0,\mathcal{K}} = \mathbb{S}^{R}_{\mathcal{K}}$,  the set $\mathbb{S}^{R}_{\mathcal{K}}$ has \emph{only one} recurrent class.\\

Given a set $\mathbb{F}$ of forbidden states, we define the set of states that are recurrent and for which the probability of transitioning into $\mathbb{F}$ is zero. 

The \emph{set of $\mathbb{F}$-safe recurrent states} $\mathbb{S}^{R}_{\mathcal{K},\mathbb{F}}$ under a control policy $\mathcal{K}$ is defined as: 
\begin{align*}
	\mathbb{S}^{R}_{\mathcal{K},\mathbb{F}}&\overset{def}{=} \Big\{ s \in \mathbb{S}^{R}_{\mathcal{K}}\text{ : } \mathcal{Q}_{\mathcal{K}} (s^+, s) = 0 \text{, } s^+ \in \mathbb{F} \Big\}.
\end{align*}

We define the \emph{{\bf maximal} set of $\mathbb{F}$-safe recurrent states} as:
\begin{align*}
	\mathbb{S}^{R}_{\mathbb{F}} \overset{def}{=} \bigcup_{\mathcal{K} \in \mathbb{K} } \mathbb{S}^{R}_{\mathcal{K},\mathbb{F}}.
\end{align*}

Finally, the \emph{set of $\mathbb{F}$-safe persistently surveilled states} $\mathbb{S}^{ps}_{s_0,\mathcal{K},\mathbb{F}}$ under a control policy $\mathcal{K}$ and initial state $s_0\in\mathbb{S}$ is defined as: 
\begin{align*}
	\mathbb{S}^{ps}_{s_0,\mathcal{K},\mathbb{F}}&\overset{def}{=} \Big\{ s \in \mathbb{S}^{ps}_{s_0,\mathcal{K}}\text{ : } \mathcal{Q}_{\mathcal{K}} (s^+, s) = 0 \text{, } s^+ \in \mathbb{F} \Big\}.
\end{align*}

\begin{rem}\label{rem2} As before, $\mathbb{S}^{ps}_{s_0,\mathcal{K},\mathbb{F}}$ is a (safe) recurrent class of $\mathcal{Q}_\mathcal{K}$.
\end{rem}
%PROBLEM STATEMENT
\subsection{Problem Statement}

We start by addressing the following problem:
\begin{problem}\label{prob1}(Maximal set of $\mathbb{F}$-safe recurrent states).
Given a set of forbidden states $\mathbb{F}$, determine:
\begin{enumerate}[(a)]
\item $\mathbb{S}^{R}_{\mathbb{F}}$; and
\item  a control policy $\mathcal{K}^*$ such that $\mathbb{S}^{R}_{\mathcal{K}^{*}} = \mathbb{S}^{R}_{\mathbb{F}}$. 
\end{enumerate}
\end{problem} 

In light of Remark \ref{rem1}, note that in order to persistently surveil all possible states, we need to determine how many robots to use and in which state they should be initialized. The following problem addresses this issue.
\begin{problem}\label{prob2} (Maximal $\mathbb{F}$-safe persistent surveillance). Given a set of forbidden states $\mathbb{F}$, determine the \emph{minimum} number of robots $r$, a control policy $\hat{\mathcal{K}}$ and a set of initial states $\{s^1, ..., s^r\}$, $r $, so that
\begin{align}
\bigcup_{i=1}^r\mathbb{S}^{ps}_{s^i,\hat{\mathcal{K}},\mathbb{F}} = \mathbb{S}^{R}_{\mathbb{F}}.
\end{align}
\end{problem}

\begin{rem} The following is a list of important comments on Problems \ref{prob1} and \ref{prob2}.
\begin{itemize}
\item There is no $\mathcal{K}$ such that the states in $\mathbb{S} \diagdown \mathbb{S}_{\mathbb{F}}^{R}$ can be $\mathbb{F}$-safe and recurrent
\item Once $r$ robots are initialized with initial states $\{s^1, ..., s^r\}$, it is guaranteed that the largest possible set of states will be visited infinitely many times without ever visiting a forbidden state.
\end{itemize}
\end{rem}

We will propose a convex optimization problem that efficiently computes $\mathbb{S}^{R}_{\mathbb{F}}$ and a control policy $\mathcal{K}^*$ such that $\mathbb{S}^{R}_{\mathcal{K}^{*}} = \mathbb{S}^{R}_{\mathbb{F}}$. We will show that the minimum number of robots $r$ required to persistently surveil $\mathbb{S}^{R}_{\mathbb{F}}$ is the number of distinct recurrent classes of the closed loop Markov chain under the computed control policy $\mathcal{K}^*$.

%SECTION: CONVEC PROGRAM
\section{Computing the Maximal Set of Recurrent States: A Convex Approach}\label{sec:recurrent}

Let $\mathbb{P}_{\mathbb{SU}}$ be the set of all probability mass functions (pmfs) with support in $\mathbb{S} \times \mathbb{U}$, and consider the following convex optimization program:
\begin{align}
	&f^*_{SU} =  \arg \max_{f_{SU} \in \mathbb{P}_\mathbb{SU}} \mathcal{H}(f_{SU})\label{eqn:opt1a}\\
	&\text{subject to: } \nonumber\\
	&\sum_{u^+ \in\mathbb{U}}f_{SU}(s^+,u^+) = \sum_{s \in \mathbb{S}, u\in\mathbb{U}}\mathcal{Q}(s^+, s, u)f_{SU}(s,u) \label{eqn:opt1b}\\ %x^+ \in \mathbb{X}\nonumber\\
	&\sum_{u \in \mathbb{U}} f_{SU}(s,u) = 0, ~~ s \in \mathbb{F}\label{eqn:opt1c}
\end{align}

where $\mathcal{H}: \mathbb{P}_{\mathbb{S}\mathbb{U}} \rightarrow \Re_{\geq 0}$ is the entropy of $f_{SU}$, and is given by
\begin{align*}
	\mathcal{H}(f_{SU}) = -\sum_{u \in \mathbb{U}}\sum_{s \in \mathbb{S}}f_{SU}(s,u)\ln\big(f_{SU}(s,u)\big)
\end{align*} where we adopt the standard convention that $0 \ln ( 0 ) = 0$.

The following proposition, which has been modified from Theorem 3.1 in \cite{earvelo}, provides a solution to Problem~\ref{prob1}.

\vspace{1mm}
%THEOREM MAX RECURRENT.
\begin{proposition}\label{prop1}
Let $\mathbb{F}$ be given, assume that (\ref{eqn:opt1a})-(\ref{eqn:opt1c}) is feasible, and that $f^*_{SU}$ is the optimal solution. In addition, adopt the marginal pmf $f^*_S(s) = \sum_{u \in \mathbb{U}} f^*_{SU}(s,u)$
and consider that $\mathcal{G}: \mathbb{U} \times \mathbb{S} \rightarrow [0,1]$ is any 
function satisfying $\sum_{u \in \mathbb{U}} \mathcal{G} (u,s)=1$ for all $s$ in $\mathbb{S}$.
The following holds:
\begin{enumerate}[(a)]
	\item $\mathbb{S}^{R}_{\mathbb{F}} = \mathbb{W}_{f^*_S}$\vspace{1mm}
	\item $\mathbb{S}^{R}_{\mathcal{K}^*} = \mathbb{S}^{R}_{\mathbb{F}}$ for $\mathcal{K}^*$ given by:
	\begin{align}\label{eqn:control_opt}
		\!\!\!\!\!\!\!\! \mathcal{K}^*(u,s) =
		\begin{cases}
		 	\frac{f^*_{SU}(s,u)}{f^*_S(s)}, & \!s \in \mathbb{W}_{f^*_S} \\
			\mathcal{G}(u,s),\! &\text{otherwise}
		\end{cases}, \quad (u,s) \in \mathbb{U}\times\mathbb{S}
	\end{align}
\end{enumerate}
where $\mathbb{W}_{f^*_S}$ is the support of $f^*_S$ and is given by $\mathbb{W}_{f^*_S} = \{s\in \mathbb{S} \text{ : }  f^*_S(s)>0\}.$
\end{proposition}
\vspace{1mm}

\paragraph*{Comments on the proof of Proposition \ref{prop1}} The proof of Proposition \ref{prop1} closely follows the proof of Theorem 3.1 in \cite{earvelo} and is omitted. However, it is important to highlight that constraint (\ref{eqn:opt1b}) enforces recurrence and constraint (\ref{eqn:opt1c}) enforces $\mathbb{F}$-safety. Moreover, note that the pmf that maximizes the entropy under convex constraints has maximal support (see Lemma 3.5 in \cite{earvelo}).

%\vspace{4mm}
%\noindent {\it Example 1}
\vspace{2mm}

\begin{ex}\label{example1}
Let $\mathbb{X} = \mathbb{Y}  = \{1,..., 5\}$, $\mathbb{O} = \{R, U, L, D\}$ and consider a robot whose action space is given by $\mathbb{U} = \{``forward", ``turn~right"\}$. Moreover, let the set of forbidden states be given by: $ \mathbb{F} = \big\{(x,y,\theta) \in \mathbb{S} \text{~:~} (x,y)  \in \{(1,1), (1,5), (5,1), (5,5), (3,3)\}\big\}$, which means the robot is prohibited from visiting the center or corner locations of the lattice. 

In order to specify $\mathcal{Q}$, we first define an auxiliary conditional pmf $\mathcal{Q}'$ defined on $\mathbb{X}' = \mathbb{Y}'  = \{1,2,3\}$ and $\mathbb{O} '= \{R, U, L, D\}$. For clarity,  $\mathcal{Q}'$ is shown graphically in Fig. \ref{fig:dyn1}, which contains the probabilities of transitioning to states shown as dark triangles given the previous state shown as a white triangle. There is uncertainty only for transitions that occur on the edge of the lattice. Since we consider dynamics that are spatially invariant, the transition probabilities for states not shown in Fig. \ref{fig:dyn1} can be computed by appropriate manipulation of the ones shown. Similarly, $\mathcal{Q}$ is constructed by appropriate expansion of $\mathcal{Q}'$.

%\vspace{-5mm}
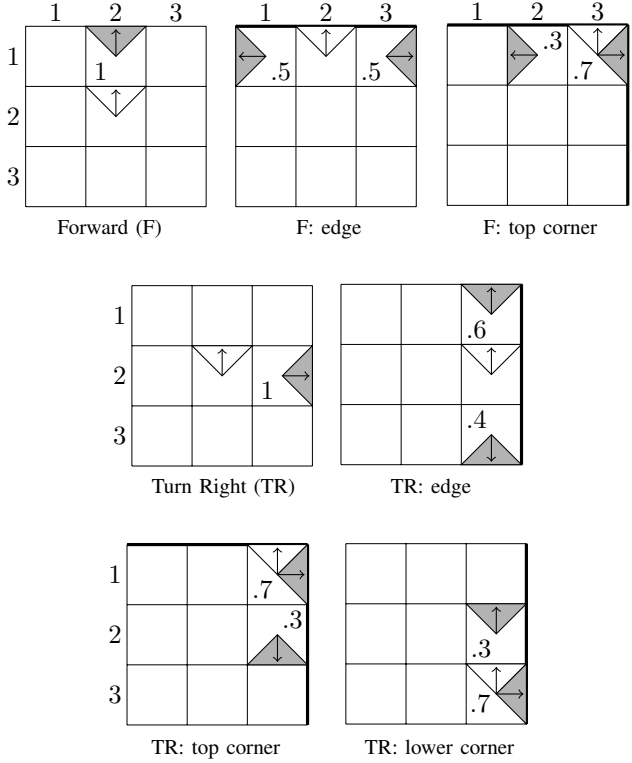
\begin{figure}[!ht]
  \captionsetup[subfigure]{labelformat=empty}
  \centering
  \subfloat[][{\hspace{4mm}\vspace{2mm} Forward (F)}]{
	\begin{tikzpicture}[inner sep=2mm, xscale = 0.8, yscale = 0.8]
	%\draw [->] (0,3) -- (3.2,3);
	%\draw [->] (0,3) -- (0,-.2);
	\draw[step=1cm, very thin] (0,0) grid (3,3);
		
	\draw[fill=black!0] (1.5,1.5)--(1,2)--(2,2)--cycle;
	%\node at (1,1.5) [circle,draw=black!50,fill=black!20] {};
	\draw[thin] [->] (1.5, 1.5) -- (1.5,1.95);
	
	\draw[fill=black!30] (1.5,2.5)--(1,3)--(2,3)--cycle;
	%\node at (1,1.5) [circle,draw=black!50,fill=black!20] {};
	\draw[thin] [->] (1.5, 2.5) -- (1.5, 2.95);
	
	\node at (1.25, 2.25) [draw=none,fill=none] {$1$};
	
	\node at (0.5, 3.23) [draw=none,fill=none] {$1$};
	\node at (1.5, 3.23) [draw=none,fill=none] {$2$};
	%\node at (1.5, 2.23) [draw=none,fill=none] {$3$};
	\node at (2.5, 3.23) [draw=none,fill=none] {$3$};
	%\node at (3.5, 2.23) [draw=none,fill=none] {$4$};
	\node at (-0.2, 2.5) [draw=none,fill=none] {$1$};
	\node at (-0.2, 0.5) [draw=none,fill=none] {$3$};
	\node at (-0.2, 1.5) [draw=none,fill=none] {$2$};
	%\node at (3.4, 3) [draw=none,fill=none] {$x$};
	%\node at (0, -0.4) [draw=none,fill=none] {$y$};
\end{tikzpicture}}~
  \subfloat[][{\hspace{2mm}F: edge}]{
	\begin{tikzpicture}[inner sep=2mm, xscale = 0.8, yscale = 0.8]
	\draw [-, very thick] (0,3) -- (3,3);
	%\draw [->] (0,3) -- (0,-.2);
	\draw[step=1cm, very thin] (0,0) grid (3,3);
	
	\draw[fill=black!0] (1.5,2.5)--(1,3)--(2,3)--cycle;
	%\node at (1,1.5) [circle,draw=black!50,fill=black!20] {};
	\draw[thin] [->] (1.5, 2.5) -- (1.5, 2.95);
	
	\draw[fill=black!30] (2.5,2.5)--(3,2)--(3,3)--cycle;
	%\node at (1,1.5) [circle,draw=black!50,fill=black!20] {};
	\draw[thin] [->] (2.5, 2.5) -- (2.95, 2.5);
	
	\draw[fill=black!30] (.5,2.5)--(0,2)--(0,3)--cycle;
	%\node at (1,1.5) [circle,draw=black!50,fill=black!20] {};
	\draw[thin] [->] (.5, 2.5) -- (0.05, 2.5);
	
	\node at (.75, 2.25) [draw=none,fill=none] {$.5$};
	\node at (2.25, 2.25) [draw=none,fill=none] {$.5$};
		
	\node at (0.5, 3.23) [draw=none,fill=none] {$1$};
	\node at (1.5, 3.23) [draw=none,fill=none] {$2$};
	\node at (2.5, 3.23) [draw=none,fill=none] {$3$};

\end{tikzpicture}}~
  \subfloat[][{\hspace{2mm}F: top corner}]{
	\begin{tikzpicture}[inner sep=2mm,, xscale = 0.8, yscale = 0.8]
	%\draw [->] (0,3) -- (3.2,3);
	%\draw [->] (0,3) -- (0,-.2);
	\draw [-, very thick] (0,3) -- (3,3);
	\draw [-, very thick] (3,3) -- (3,0);
	\draw[step=1cm, very thin] (0,0) grid (3,3);
	
	\draw[fill=black!0] (2.5,2.5)--(2,3)--(3,3)--cycle;
	%\node at (1,1.5) [circle,draw=black!50,fill=black!20] {};
	\draw[thin] [->] (2.5, 2.5) -- (2.5, 2.95);

	\draw[fill=black!30] (1.5,2.5)--(1,2)--(1,3)--cycle;
	%\node at (1,1.5) [circle,draw=black!50,fill=black!20] {};
	\draw[thin] [->] (1.5, 2.5) -- (1.05, 2.5);
	
	\draw[fill=black!30] (2.5,2.5)--(3,3)--(3,2)--cycle;
	%\node at (1,1.5) [circle,draw=black!50,fill=black!20] {};
	\draw[thin] [->] (2.5, 2.5) -- (2.95, 2.5);
	
	\node at (1.75, 2.75) [draw=none,fill=none] {$.3$};
	\node at (2.25, 2.25) [draw=none,fill=none] {$.7$};
		
	\node at (0.5, 3.23) [draw=none,fill=none] {$1$};
	\node at (1.5, 3.23) [draw=none,fill=none] {$2$};
	\node at (2.5, 3.23) [draw=none,fill=none] {$3$};

\end{tikzpicture}}\quad
  \subfloat[][{\hspace{6mm}Turn Right (TR)}]{
	\begin{tikzpicture}[inner sep=2mm, xscale = 0.8, yscale = 0.8]
	%\draw [->] (0,3) -- (3.2,3);
	%\draw [->] (0,3) -- (0,-.2);
	\draw[step=1cm, very thin] (0,0) grid (3,3);
	
	\draw[fill=black!0] (1.5,1.5)--(1,2)--(2,2)--cycle;
	%\node at (1,1.5) [circle,draw=black!50,fill=black!20] {};
	\draw[thin] [->] (1.5, 1.5) -- (1.5,1.95);
	
	\draw[fill=black!30] (2.5,1.5)--(3,1)--(3,2)--cycle;
	%\node at (1,1.5) [circle,draw=black!50,fill=black!20] {};
	\draw[thin] [->] (2.5, 1.5) -- (2.95, 1.5);
	
	\node at (2.25, 1.25) [draw=none,fill=none] {$1$};

	\node at (-0.2, 2.5) [draw=none,fill=none] {$1$};
	\node at (-0.2, 0.5) [draw=none,fill=none] {$3$};
	\node at (-0.2, 1.5) [draw=none,fill=none] {$2$};

\end{tikzpicture}}~
  \subfloat[][{\hspace{1mm}\vspace{2mm}TR: edge}]{
	\begin{tikzpicture}[inner sep=2mm, xscale = 0.8, yscale = 0.8]
	%\draw [->] (0,3) -- (3.2,3);
	%\draw [->] (0,3) -- (0,-.2);
		%\draw [-, very thick] (0,3) -- (3,3);
	\draw [-, very thick] (3,3) -- (3,0);
	\draw[step=1cm, very thin] (0,0) grid (3,3);
	
	\draw[fill=black!0] (2.5,1.5)--(2,2)--(3,2)--cycle;
	%\node at (1,1.5) [circle,draw=black!50,fill=black!20] {};
	\draw[thin] [->] (2.5, 1.5) -- (2.5,1.95);
	
	\draw[fill=black!30] (2.5, 0.5)--(3,0)--(2,0)--cycle;
	%\node at (1,1.5) [circle,draw=black!50,fill=black!20] {};
	\draw[thin] [->] (2.5, 0.5) -- (2.5,0.05);
	
	\draw[fill=black!30] (2.5,2.5)--(2,3)--(3,3)--cycle;
	%\node at (1,1.5) [circle,draw=black!50,fill=black!20] {};
	\draw[thin] [->] (2.5, 2.5) -- (2.5, 2.95);
	
	\node at (2.25, 2.25) [draw=none,fill=none] {$.6$};
	\node at (2.25, .75) [draw=none,fill=none] {$.4$};

\end{tikzpicture}}\qquad
  \subfloat[][{\hspace{4mm}TR: top corner}]{
	\begin{tikzpicture}[inner sep=2mm, xscale = 0.8, yscale = 0.8]
	%\draw [->] (0,3) -- (3.2,3);
	%\draw [->] (0,3) -- (0,-.2);
		\draw [-, very thick] (0,3) -- (3,3);
	\draw [-, very thick] (3,3) -- (3,0);
	\draw[step=1cm, very thin] (0,0) grid (3,3);
	
	\draw[fill=black!0] (2.5,2.5)--(2,3)--(3,3)--cycle;
	%\node at (1,1.5) [circle,draw=black!50,fill=black!20] {};
	\draw[thin] [->] (2.5, 2.5) -- (2.5, 2.95);
	
	\draw[fill=black!30] (2.5,2.5)--(3,3)--(3,2)--cycle;
	%\node at (1,1.5) [circle,draw=black!50,fill=black!20] {};
	\draw[thin] [->] (2.5, 2.5) -- (2.95, 2.5);
	
	\draw[fill=black!30] (2.5,1.5)--(2,1)--(3,1)--cycle;
	%\node at (1,1.5) [circle,draw=black!50,fill=black!20] {};
	\draw[thin] [->] (2.5, 1.5) -- (2.5, 1.05);
	
	\node at (2.25, 2.25) [draw=none,fill=none] {$.7$};
	\node at (2.75, 1.75) [draw=none,fill=none] {$.3$};
	
	\node at (-0.2, 2.5) [draw=none,fill=none] {$1$};
	\node at (-0.2, 0.5) [draw=none,fill=none] {$3$};
	\node at (-0.2, 1.5) [draw=none,fill=none] {$2$};

	\node at (2.25, 0.25) [draw=none,fill=none] {};
\end{tikzpicture}}~
  \subfloat[][{\hspace{2mm}TR: lower corner}]{
	\begin{tikzpicture}[inner sep=2mm, xscale = 0.8, yscale = 0.8]
	%\draw [->] (0,3) -- (3.2,3);
	%\draw [->] (0,3) -- (0,-.2);
	%	\draw [-, very thick] (0,3) -- (3,3);
	\draw [-, very thick] (3,3) -- (3,0);
	\draw[step=1cm, very thin] (0,0) grid (3,3);
	
	\draw[fill=black!0] (2.5,0.5)--(2,1)--(3,1)--cycle;
	%\node at (1,1.5) [circle,draw=black!50,fill=black!20] {};
	\draw[thin] [->] (2.5, 0.5) -- (2.5, 0.95);
	
	\draw[fill=black!30] (2.5,0.5)--(3,1)--(3,0)--cycle;
	%\node at (1,1.5) [circle,draw=black!50,fill=black!20] {};
	\draw[thin] [->] (2.5, 0.5) -- (2.95, 0.5);
	
	\draw[fill=black!30] (2.5,1.5)--(2,2)--(3,2)--cycle;
	%\node at (1,1.5) [circle,draw=black!50,fill=black!20] {};
	\draw[thin] [->] (2.5, 1.5) -- (2.5, 1.95);
	
	\node at (2.25, 1.25) [draw=none,fill=none] {$.3$};
	\node at (2.23, 0.35) [draw=none,fill=none] {$.7$};

\end{tikzpicture}}
  \caption{Graphical representation of some transitions in $\mathcal{Q}'$.}
  \label{fig:dyn1}
\end{figure}

We use \cite{cvx} to solve (\ref{eqn:opt1a})-(\ref{eqn:opt1c}) and use Proposition \ref{prop1} to compute  $\mathbb{S}^{R}_{\mathbb{F}}$ and a control policy $\mathcal{K}^*$ such that $\mathbb{S}^{R}_{\mathcal{K}^*} = \mathbb{S}^{R}_{\mathbb{F}}$. The set $\mathbb{S}^{R}_{\mathbb{F}}$ can be seen in Fig. \ref{fig:simple1}, where the areas in red represent the states in $\mathbb{F}$, and the triangles in blue represent the states in $\mathbb{S}^{R}_{\mathbb{F}}$. The control $\mathcal{K}^*$, computed using (\ref{eqn:control_opt}), has been omitted due to space constraints.
\begin{figure}[t]
  \centering
 \includegraphics[trim=3cm 1cm 3cm 0cm, clip = true, width=0.35\textwidth]{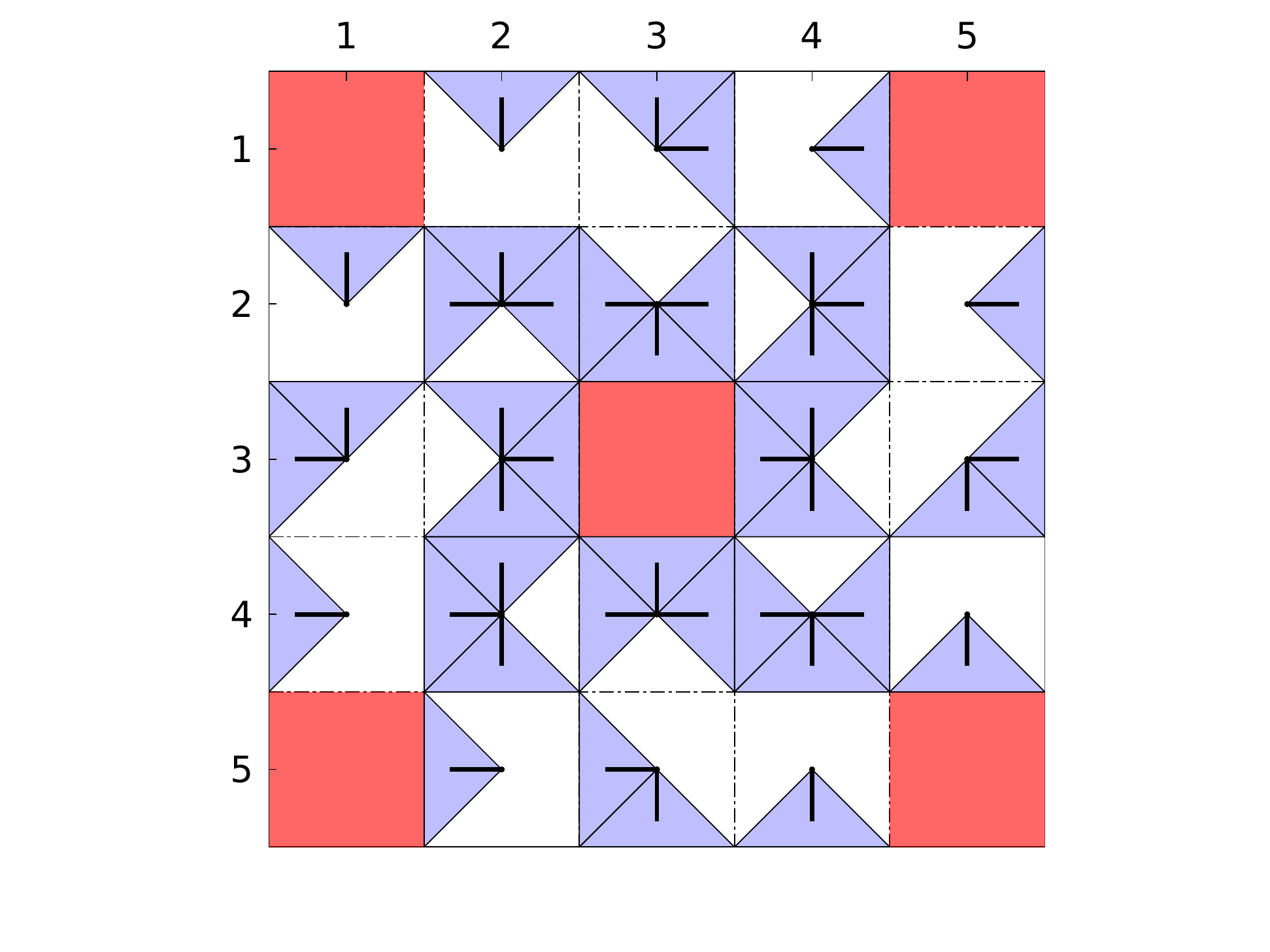}
  \caption{Depiction of $\mathbb{S}^{R}_{\mathbb{F}}$ in blue. The red areas represent the forbidden states.}
\label{fig:simple1}
\end{figure}
\end{ex}

%SECTION
%\vspace{-10mm}
\section{Maximal Persistent Surveillance and \\Robot Deployment}\label{sec:deploy}
In this section, we provide a solution to Problem \ref{prob2}, which seeks the minimum number $r$ of robots, a control policy $\hat{\mathcal{K}}$ and a set of initial states $\{s^1, ..., s^r\}$ so that $\bigcup_{i=1}^r\mathbb{S}^{ps}_{s^i,\hat{\mathcal{K}},\mathbb{F}} = \mathbb{S}^{R}_{\mathbb{F}}$. 

In light of a previous remark, recall that the set of $\mathbb{F}$-safe persistently surveilled states $\mathbb{S}^{ps}_{s_0,\mathcal{K},\mathbb{F}}$ is a \emph{recurrent class} of  $\mathcal{Q}_\mathcal{K}$. In practice, this means that when a robot with initial state $S_0 = s_0$ applies control policy $\mathcal{K}$, it is guaranteed that: \begin{itemize} 
\item the robot will never leave $\mathbb{S}^{ps}_{s_0,\mathcal{K},\mathbb{F}}$;
\item every state in $\mathbb{S}^{ps}_{s_0,\mathcal{K},\mathbb{F}}$ will be visited infinitely many times;
\item states in $\mathbb{F}$ will never be visited.
\end{itemize}

To find all the (safe) recurrent classes in $\mathbb{S}^{R}_{\mathcal{K},\mathbb{F}}$, flood-fill-type algorithms may be used, where the graph of $\mathcal{Q}_\mathcal{K}$ is traversed, either in a depth-first or breath-first manner. An edge from $s$ to $s^+$ of the graph of $\mathcal{Q}_{\mathcal{K}}$ exists if and only if $\mathcal{Q}_\mathcal{K}(s^+, s)>0$ holds. Note that states in $\mathbb{S} \diagdown (\mathbb{S}^{R}_{\mathcal{K},\mathbb{F}} \cup \mathbb{F})$ do not need to be searched. \\

Given $\mathbb{F}$ and a control policy $\mathcal{K}$, let $n_\mathcal{K}$ be the number of distinct recurrent classes of $\mathcal{Q}_{\mathcal{K}}$, and note that the following holds: $$\bigcup_{i=1}^{n_\mathcal{K}}\mathbb{S}^{ps}_{s^i,{\mathcal{K}},\mathbb{F}} = \mathbb{S}^{R}_{\mathcal{K},\mathbb{F}},$$ 

\noindent where $\{s^1, ..., s^{n_{\mathcal{K}}}\}$ is a set of initial states, and $\big\{\mathbb{S}^{ps}_{s^i,{\mathcal{K}},\mathbb{F}}\big\}_{i=1}^{n_\mathcal{K}}$ are distinct recurrent classes. 
\vspace{2mm}

We define the set of all admissible control policies whose $\mathbb{F}$-safe set of recurrent states are maximal to be:
$$\mathbb{K}^{R}_\mathbb{F} = \big\{ \mathcal{K} \in \mathbb{K} \text{~:~} \mathbb{S}^{R}_{\mathcal{K},\mathbb{F}} =\mathbb{S}^{R}_{\mathbb{F}}\big\},$$
and note that in order to solve Problem \ref{prob2}, we must:
\begin{itemize}
\item find a control policy $\hat{\mathcal{K}}$ in $\mathbb{K}^{R}_\mathbb{F}$ such that $n_{\hat{\mathcal{K}}}\leq n_{\mathcal{K}}$ for all $\mathcal{K}$ in $\mathbb{K}^{R}_\mathbb{F}$. Note that $n_{\hat{\mathcal{K}}}$ is the minimum number of robots needed for maximal persistent surveillance.
\item identify the recurrent classes in $\mathbb{S}^{R}_{\hat{\mathcal{K}},\mathbb{F}}$ (by exploring the graph of $\mathcal{Q}_{\hat{\mathcal{K}}}$); and 
\item select one (\emph{any}) state from each of the recurrent classes to compose the set of initial states 
\end{itemize}

Note that the control policy $\mathcal{K}^*$ given in (\ref{eqn:control_opt}) is a candidate for maximal persistent surveillance since $\mathbb{S}^{R}_{{\mathcal{K}^*},\mathbb{F}} =\mathbb{S}^{R}_{\mathbb{F}}$. The following proposition will show that $n_{\mathcal{K}^*}\leq n_{\mathcal{K}}$ for all $\mathcal{K}$ in $\mathbb{K}^{R}_\mathbb{F}$.

\begin{proposition}\label{prop2} Let $\mathbb{F}$ be given, and take $\mathcal{K}^*$ to be the control policy in (\ref{eqn:control_opt}). The following holds: $$n_{\mathcal{K}^*}\leq n_{\mathcal{K}},~~~~~~~\mathcal{K} \in \mathbb{K}^{R}_\mathbb{F}.$$
\end{proposition}
\begin{proof}
Suppose there exists a control policy $\bar{\mathcal{K}}$ in $\mathbb{K}^{R}_\mathbb{F}$ such that $n_{\bar{\mathcal{K}}} < n_{\mathcal{K}^*}$. Since $\bar{\mathcal{K}}$ belongs to $\mathbb{K}^{R}_\mathbb{F}$, there must exist a control policy $\tilde{\mathcal{K}}$ with the same sparsity pattern as $\bar{\mathcal{K}}$ and a pmf $\tilde{f}_{SU}$ in $\mathbb{P}_{\mathbb{SU}}$ that satisfies (\ref{eqn:opt1b}) and (\ref{eqn:opt1c}) for which:
\begin{align*}
		\!\!\!\!\!\!\!\! \tilde{\mathcal{K}}(u,s) =
		\begin{cases}
		 	\frac{\tilde{f}_{SU}(s,u)}{\tilde{f_S}(s)}, & \!s \in \mathbb{S}^{R}_{\mathbb{F}} \\
			\mathcal{G}(u,s),\! &\text{otherwise}
		\end{cases}, \quad (u,s) \in \mathbb{U}\times\mathbb{S}
	\end{align*}
where $\tilde{f}_S(s) = \sum_{u \in \mathbb{U}} \tilde{f}_{SU}(s,u)$ and $\mathcal{G}: \mathbb{U} \times \mathbb{S} \rightarrow [0,1]$.

Since $\mathcal{Q}_{\bar{\mathcal{K}}}$ has fewer recurrent classes than $\mathcal{Q}_{\mathcal{K}^*}$, there must exist a pair $(s,u)$ in $\mathbb{S}^{R}_{\mathbb{F}} \times\mathbb{U}$ for which $\bar{\mathcal{K}}(s,u)>0$ and $\mathcal{K}^*(s,u)=0$ holds. Since $\bar{\mathcal{K}}$ and $\tilde{\mathcal{K}}$ have the same sparsity pattern, it holds that $\tilde{\mathcal{K}}(s,u)>0$. Therefore, it must be that $\tilde{f}_{SU}(s,u)>0$ and $f^*_{SU}(s,u)=0$. In other words, the support of $\tilde{f}_{SU}$ is not contained in the support of  $f^*_{SU}$, which is a contradiction by Lemma 3.5 in \cite{earvelo}. 
\end{proof}

\begin{rem}
Suppose we change the objective function in (\ref{eqn:opt1a}) to $\mathcal{H}(f_{S})$ and add the following constraint: $f_S(s) = \sum_{u \in \mathbb{U}}f_{SU}(u, s)$. Note that an appropriate modification of Proposition~\ref{prop1} would enable us to find $\mathbb{S}^{R}_{\mathbb{F}}$ and an associated control policy (i.e., solve Problem \ref{prop1}), with the added benefit that the modified convex program would be computationally less intensive (since fewer calls to the entropy function would be required). However, maximizing the entropy of the marginal distribution $f_{S}$ would not solve the problem of maximal persistent surveillance since Proposition \ref{prop2} would not apply.
\end{rem}
\vspace{2mm}
%\noindent {\it Example 2}
%\vspace{2mm}

\begin{ex}\label{example2}
Consider again the example described in Example \ref{example1}. By exploring the graph of $\mathcal{Q}_{\mathcal{K}^*}$, we conclude that only one robot is required to perform maximal persistent surveillance (i.e., $\mathbb{S}^{R}_\mathbb{F}$ contains only one recurrent class). Any state in $\mathbb{S}^{R}_\mathbb{F}$ may be selected as the robot's initial state.

Suppose that we now change the set of forbidden states to include location $(4,3)$ (i.e., let $\mathbb{F} = \big\{(x,y,\theta)\in \mathbb{S} \text{~:~} (x,y)  \in \{(1,1), (1,5), (5,1), (5,5), (3,3), (4,3)\}\big\}$). Re-solving (\ref{eqn:opt1a})-(\ref{eqn:opt1c}), applying Propositions \ref{prop1} and \ref{prop2}, and searching the graph of the closed loop Markov chain, we conclude that at least three robots are required to perform maximal persistent surveillance of $\mathbb{S}^{R}_\mathbb{F}$ (see Fig. \ref{fig:simple3}). Any state from each recurrent class may be used as initial states, so we can chose the set of initial states to be: $\big\{(1,2,U), (2,1,U), (2,4,U)\big\}$. Note that the set $\mathbb{S}^{R}_\mathbb{F}$ is now smaller than in the previous example ($34$ vs. $40$ states).

\begin{figure}[!tb]
  \captionsetup[subfigure]{labelformat=empty}
  \centering
  \subfloat[][{\hspace{4mm} $\mathbb{S}^{R}_{\mathbb{F}}$ }]{ 
 \includegraphics[trim=3.6cm .9cm 3.35cm 0cm, clip = true, width=0.238\textwidth]{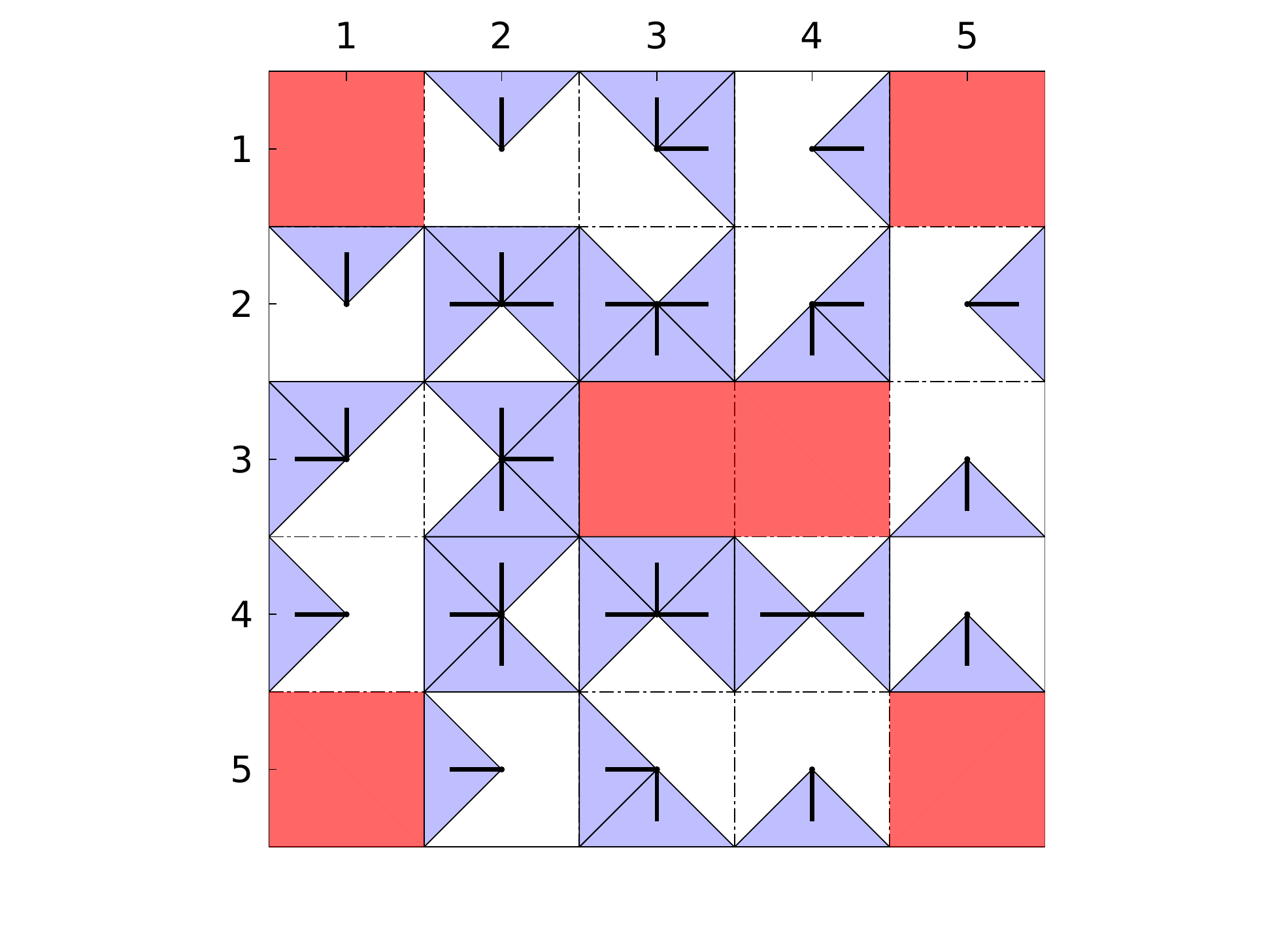}
}
  \subfloat[][{\hspace{4mm}$\mathbb{S}^{ps}_{s^1, \mathcal{K}^*,\mathbb{F}},~~s^1 = (1,2,U)$}]{ 
 \includegraphics[trim=4.1cm .9cm 3.35cm 0cm, clip = true, width=0.23\textwidth]{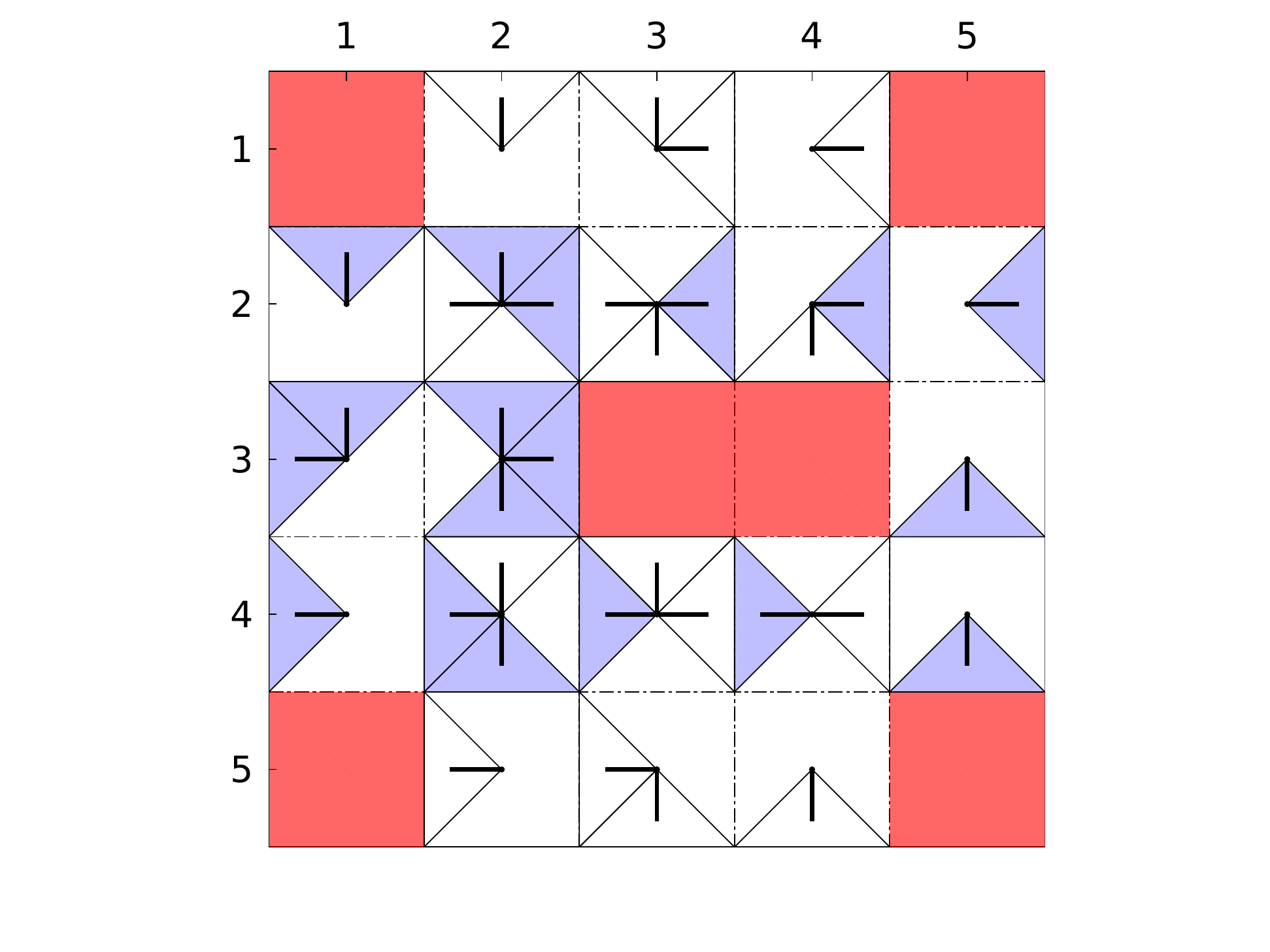}
}\qquad
  \subfloat[][{\hspace{4mm} $\mathbb{S}^{ps}_{s^2, \mathcal{K}^*,\mathbb{F}},~~s^2 = (2,1,U)$ }]{ 
 \includegraphics[trim=3.6cm .9cm 3.35cm 0cm, clip = true, width=0.238\textwidth]{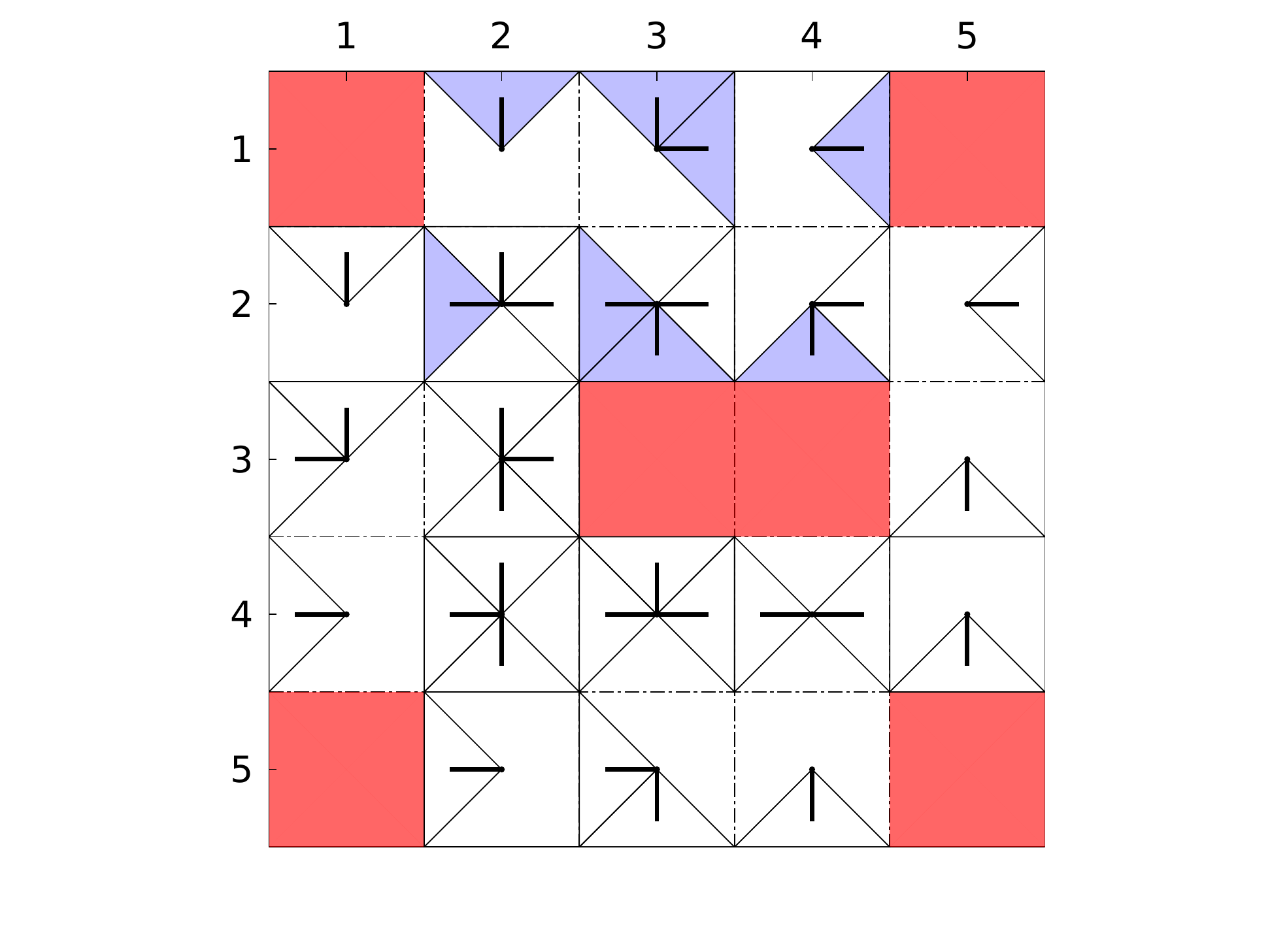}
}
  \subfloat[][{\hspace{4mm}\ $\mathbb{S}^{ps}_{s^2, \mathcal{K}^*,\mathbb{F}},~~s^3 = (2,4,U)$}]{ 
 \includegraphics[trim=4.1cm .9cm 3.35cm 0cm, clip = true, width=0.23\textwidth]{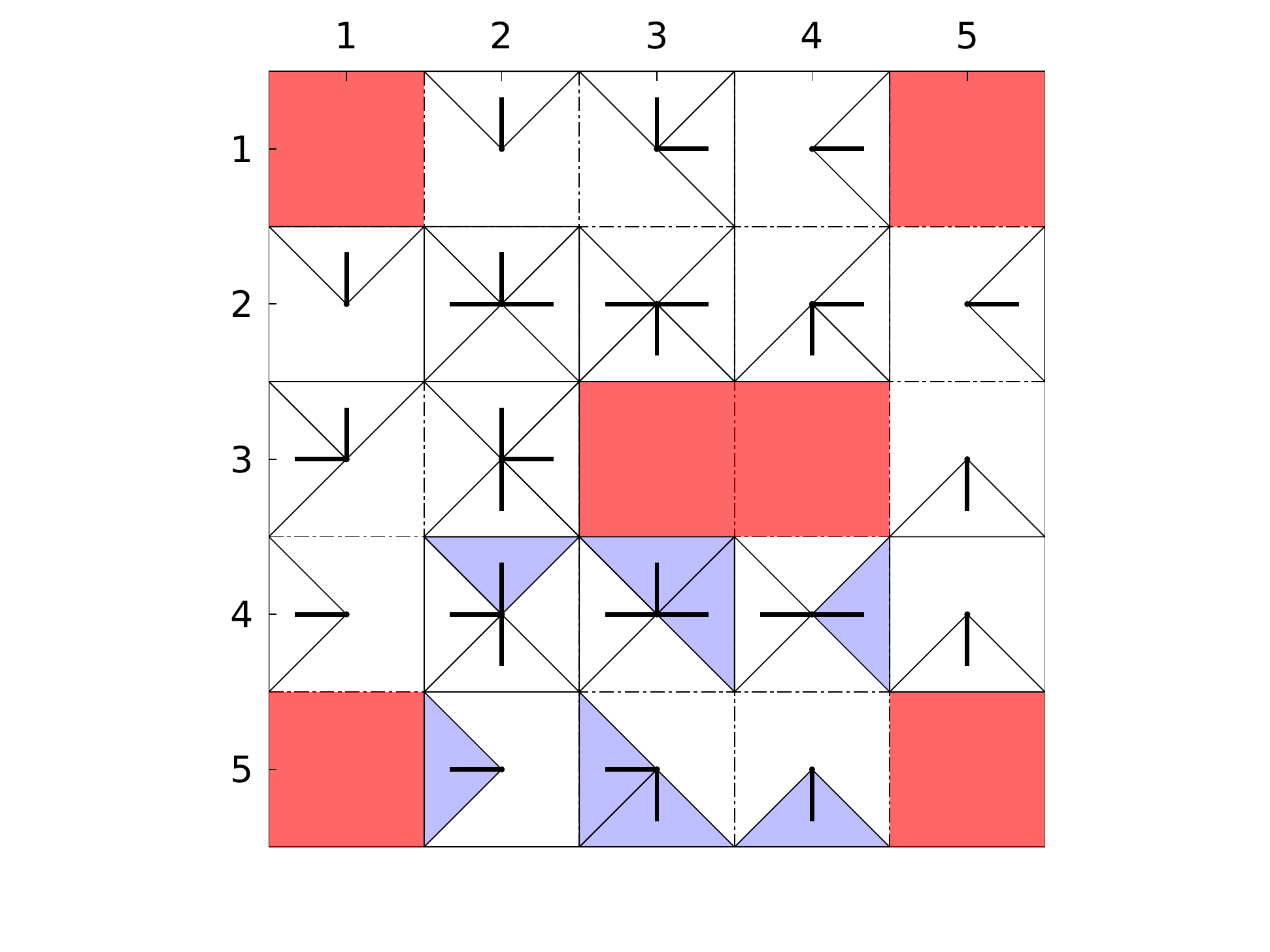}
}
  \caption{Top left: maximal set of recurrent states $\mathbb{S}^{R}_{\mathbb{F}}$ (in blue). Others: three recurrent classes whose union is $\mathbb{S}^{R}_{\mathbb{F}}$.}
  \label{fig:simple3}
\end{figure}

\end{ex}

\section{Limiting Behavior and Other Constraints}\label{sec:limiting}
We define $\mathcal{T}_\mathcal{K}$, the long term proportion of time the robot, under control policy $\mathcal{K}$, visits state $s$ in $\mathbb{S}$ having started at state $s_0$, to be: $$\mathcal{T}_\mathcal{K}(s,s_0) \overset{def}{=} \lim_{k\rightarrow \infty} \frac{1}{k}\sum_{i =1}^k\mathcal{I}\big(S_i=s, S_0 = s_0\big),$$
where $\mathcal{I}$ is the indicator function. 

\subsection{Limiting Behavior with One Recurrent Class}
Given a forbidden set $\mathbb{F}$, and let $f_{SU}^*$ be the optimal solution to (\ref{eqn:opt1a})-(\ref{eqn:opt1c}), and $\mathcal{K}^*$ be the control policy computed in (\ref{eqn:control_opt}), and suppose $\mathbb{S}^R_{\mathcal{K}^*,\mathbb{F}}$ has only one recurrent class. For any initial state $s_0$ in $\mathbb{S}^R_{\mathbb{F}}$, the following holds with probability one:
\begin{align}\label{eqn:limit}
\mathcal{T}_{\mathcal{K}^*}(s,s_0) = f^*_S(s),
\end{align}
were ${f}^*_S(s) = \sum_{u \in \mathbb{U}} {f}^*_{SU}(s,u)$. Since we have not imposed aperiodicity on $\mathcal{Q}_{\mathcal{K}^*}$, we cannot state stronger convergence. However, equation (\ref{eqn:limit}) still tells us valuable information regarding the limiting behavior of the robot.
\vspace{2mm}

Note that the pmf that maximizes the entropy is ``as uniform as possible" (in fact, when unconstrained, the pmf that maximizes the entropy is uniform.). However, additional convex constraints can be added to our formulation in order to shape the distribution of the optimal pmf and, thus, influence the limiting behavior of the robot. 

Consider the following constraint:
\begin{align}\label{eqn:newconst}
\sum_{(x,y)\in\mathbb{D},~\theta \in \mathbb{O},~u \in \mathbb{U}} f_{SU}((x,y,\theta),u)>\alpha,
\end{align}
where $\mathbb{D}\subset\mathbb{X}\times\mathbb{Y}$ is a region of the lattice. The set $\mathbb{D}$ can be interpreted as a region of high interest that should be surveilled more often. Suppose the convex program (\ref{eqn:opt1a})-(\ref{eqn:opt1c}) and (\ref{eqn:newconst}) is feasible, that $f^{**}_{SU}$ is the optimal solution and $\mathcal{K}^{**}$ is the associated control policy. The following holds for any $s_0$ in $\mathbb{S}^R_{\mathbb{F}}$ with probability one:
\begin{align*}
\sum_{(x,y) \in \mathbb{D}, ~\theta \in \mathbb{O}} \mathcal{T}_{\mathcal{K}^{**}}\big((x,y,\theta), s_0\big) > \alpha.
\end{align*}

\vspace{2mm}
%\noindent {\it Example 3}
%\vspace{2mm}

\begin{ex}\label{example3}
Let $\mathbb{X} = \mathbb{Y}  = \{1,...,10\}$, $\mathbb{O} =  \{R, U, L, D\}$, and consider again a robot whose action space is given by $\mathbb{U} = \{``Forward", ``Turn~Right"\}$. The dynamics $\mathcal{Q}$ are similar to what was used in Examples \ref{example1} and \ref{example2}, except that we add uncertainty to the transition of states that lie in the interior of the grid (see Fig. \ref{fig:dyn2}). The probabilities for states on the edge of the grid are the same as before (see Fig.\ref{fig:dyn1}).

\begin{figure}[!htb]
  \captionsetup[subfigure]{labelformat=empty}
  \centering
  \subfloat[][{\hspace{4mm}\vspace{2mm} Forward}]{
	\begin{tikzpicture}[inner sep=2mm, xscale = 0.8, yscale = 0.8]
	%\draw [->] (0,3) -- (3.2,3);
	%\draw [->] (0,3) -- (0,-.2);
	\draw[step=1cm, very thin] (0,0) grid (3,3);
		
	\draw[fill=black!0] (1.5,1.5)--(1,2)--(2,2)--cycle;
	%\node at (1,1.5) [circle,draw=black!50,fill=black!20] {};
	\draw[thin] [->] (1.5, 1.5) -- (1.5,1.95);
	
	\draw[fill=black!30] (1.5,2.5)--(1,3)--(2,3)--cycle;
	%\node at (1,1.5) [circle,draw=black!50,fill=black!20] {};
	\draw[thin] [->] (1.5, 2.5) -- (1.5, 2.95);
	
	\draw[fill=black!30] (1.5,2.5)--(1,3)--(1,2)--cycle;
	%\node at (1,1.5) [circle,draw=black!50,fill=black!20] {};
	\draw[thin] [->] (1.5, 2.5) -- (1.05, 2.5);
	
	\draw[fill=black!30] (1.5,2.5)--(2,2)--(2,3)--cycle;
	%\node at (1,1.5) [circle,draw=black!50,fill=black!20] {};
	\draw[thin] [->] (1.5, 2.5) -- (1.95, 2.5);
	
	%\node at (1.25, 2.25) [draw=none,fill=none] {$1$};
	
		\node at (.75, 2.25) [draw=none,fill=none] {$.2$};
		\node at (2.25, 2.25) [draw=none,fill=none] {$.2$};
	
	\node at (0.5, 3.23) [draw=none,fill=none] {$1$};
	\node at (1.5, 3.23) [draw=none,fill=none] {$2$};
	%\node at (1.5, 2.23) [draw=none,fill=none] {$3$};
	\node at (2.5, 3.23) [draw=none,fill=none] {$3$};
	%\node at (3.5, 2.23) [draw=none,fill=none] {$4$};
	\node at (-0.2, 2.5) [draw=none,fill=none] {$1$};
	\node at (-0.2, 0.5) [draw=none,fill=none] {$3$};
	\node at (-0.2, 1.5) [draw=none,fill=none] {$2$};
	%\node at (3.4, 3) [draw=none,fill=none] {$x$};
	%\node at (0, -0.4) [draw=none,fill=none] {$y$};
\end{tikzpicture}}~
    \subfloat[][{\hspace{3mm}Turn Right}]{
	\begin{tikzpicture}[inner sep=2mm, xscale = 0.8, yscale = 0.8]

	\draw[step=1cm, very thin] (0,0) grid (3,3);
	
	\draw[fill=black!0] (1.5,1.5)--(1,2)--(2,2)--cycle;
	\draw[thin] [->] (1.5, 1.5) -- (1.5,1.95);
	
	\draw[fill=black!30] (2.5,1.5)--(3,1)--(3,2)--cycle;

	\draw[thin] [->] (2.5, 1.5) -- (2.95, 1.5);
	
	\draw[fill=black!30] (2.5,2.5)--(3,2)--(3,3)--cycle;

	\draw[thin] [->] (2.5, 2.5) -- (2.95, 2.5);

	\node at (2.25, 1.25) [draw=none,fill=none] {$.7$};
		\node at (2.25, 2.25) [draw=none,fill=none] {$.3$};

	\node at (0.5, 3.23) [draw=none,fill=none] {$1$};
	
	\node at (1.5, 3.23) [draw=none,fill=none] {$2$};
	\node at (2.5, 3.23) [draw=none,fill=none] {$3$};

\end{tikzpicture}}~
  \caption{Graphical representation of some transitions in $\mathcal{Q}'$.}
  \label{fig:dyn2}
\end{figure}
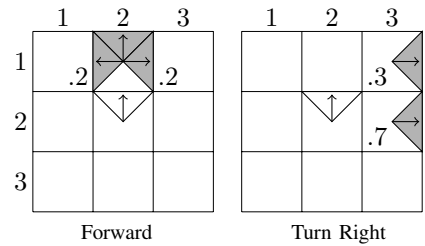

Consider the set of forbidden state $\mathbb{F}~ = \big\{(x,y,\theta) ~\in ~\mathbb{S} \text{~:} \\(x,y) \! \in \! \{(2,2), \!(2,3),\! (3,2),\! (3,3),\! (8,8),\! (8,9),\! (9,8), \!(9,9)\}\big\}$. We solve (\ref{eqn:opt1a})-(\ref{eqn:opt1c}) using the tool in \cite{cvx}. In Fig. \ref{fig:limit1}, each state that belongs in $\mathbb{S}^R_{\mathbb{F}}$ is shown in blue, where the darker the blue, the higher the value of $f^*_S$. Note that the distribution is relatively uniform.

\begin{figure}[ht!]
  \centering
 \includegraphics[trim=3cm 1cm 3cm 0cm, clip = true, width=0.365\textwidth]{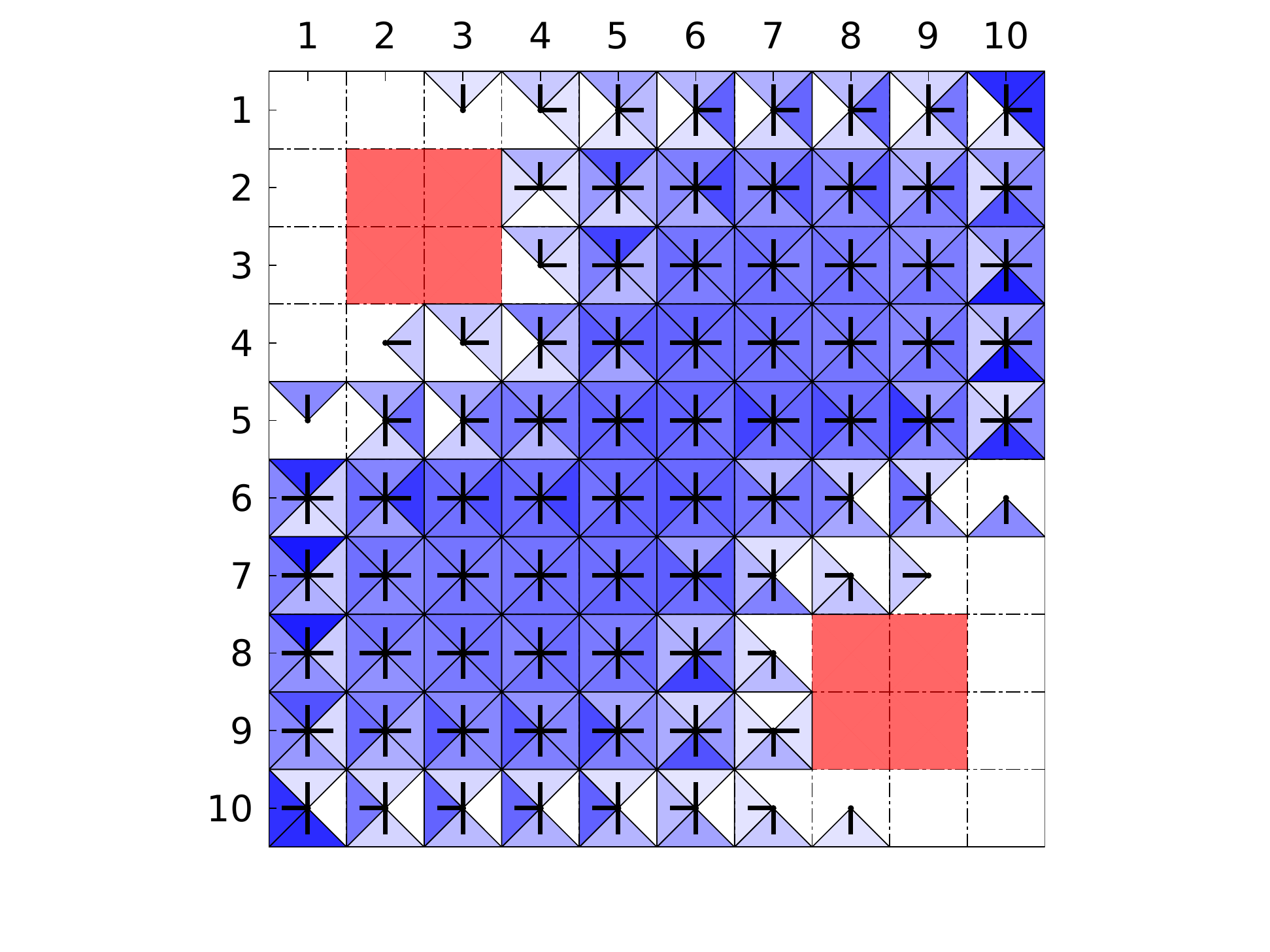}
\caption{Depiction of $\mathbb{S}^{R}_{\mathbb{F}}$ in blue. Darker blue indicates a higher value for $f^*_S$.}
\label{fig:limit1}
\end{figure}

Consider now $\mathbb{D} = \big\{(x,y) ~\in ~\mathbb{X}\times\mathbb{Y} \text{~:~} 3\leq x,y \leq 8\big\}$, and let $\alpha = 0.75$. We solve (\ref{eqn:opt1a})-(\ref{eqn:opt1c}) and (\ref{eqn:newconst}). The result can be seen in Fig. \ref{fig:limit2}.

\begin{figure}[t!]
  \centering
 \includegraphics[trim=3cm 1cm 3cm 0cm, clip = true, width=0.365\textwidth]{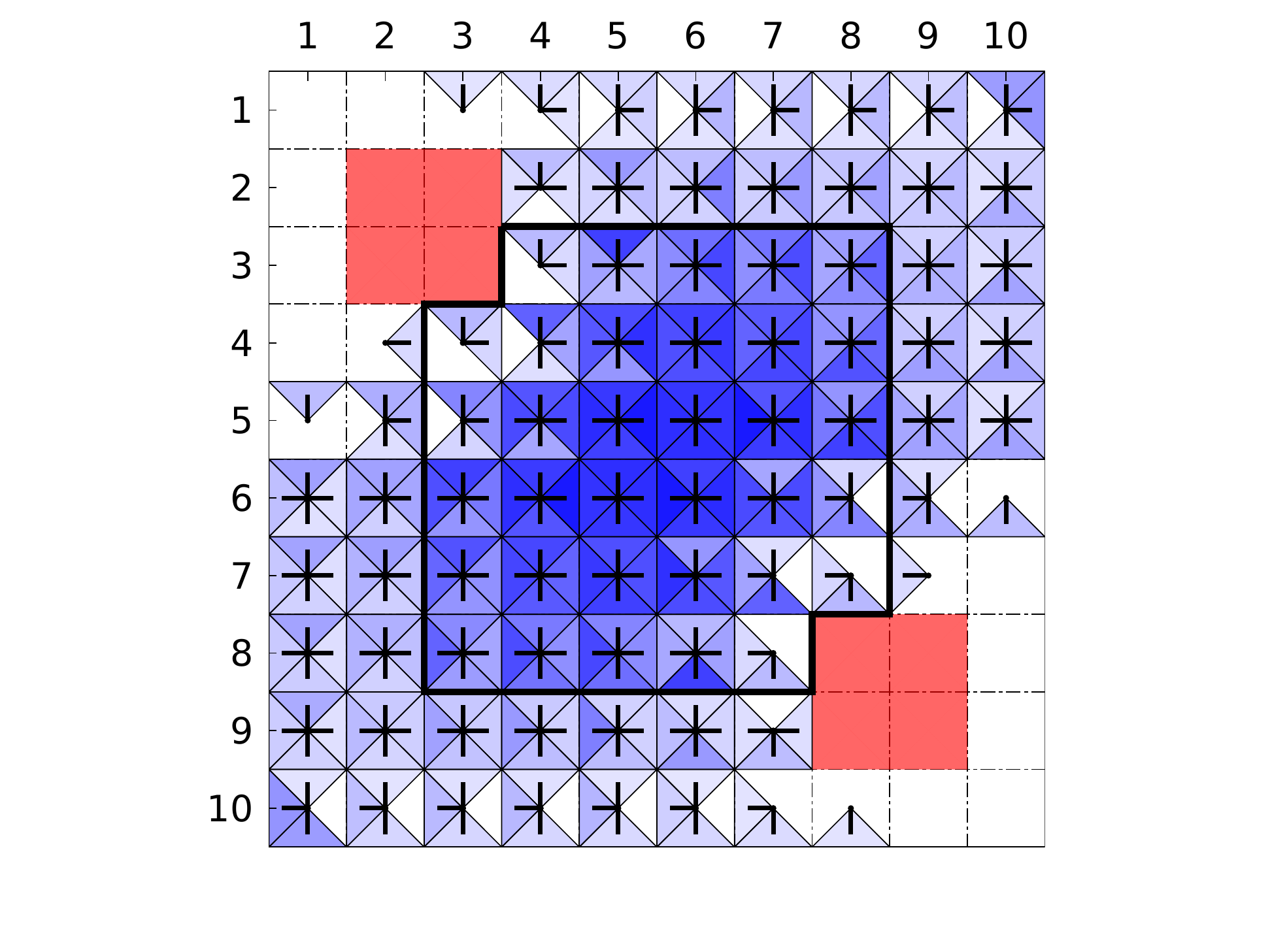}
\caption{Depiction of $\mathbb{S}^{R}_{\mathbb{F}}$ in blue with additional constraint (\ref{eqn:newconst}). Darker blue indicates a higher value for $f^{**}_S$.}
\label{fig:limit2}
\end{figure}
\end{ex}

\subsection{Limiting Behavior with Multiple Recurrent Classes}

Consider again $f_{SU}^*$  and $\mathcal{K}^*$ as before, and, without loss of generality, let $\mathbb{S}^R_{\mathcal{K}^*,\mathbb{F}}$ have two recurrent classes with initial states $s^1$ and $s^2$ \big(i.e., $\mathbb{S}^{ps}_{s^1,\mathcal{K}^*,\mathbb{F}} \cup \mathbb{S}^{ps}_{s^2,\mathcal{K}^*,\mathbb{F}} = \mathbb{S}^{R}_{\mathcal{K}^*,\mathbb{F}}$\big). For any initial state $s_0$ in $\mathbb{S}^{ps}_{s^1,\mathcal{K}^*,\mathbb{F}}$ \big(equiv., $\mathbb{S}^{ps}_{s^2,\mathcal{K}^*,\mathbb{F}}\big)$, the following holds with probability one:
\begin{align}\label{eqn:limit2}
\mathcal{T}_{\mathcal{K}^*}(s,s_0) = \frac{f^*_S(s)}{\beta},
\end{align}
where $\beta \!=\! \sum_{s \in \mathbb{S}^{ps}_{s^1,\mathcal{K}^*,\mathbb{F}}}f^*_S(s)$ \big(equiv. $\beta \!=\! \sum_{s \in \mathbb{S}^{ps}_{s^2,\mathcal{K}^*,\mathbb{F}}}f^*_S(s)$\big). 
\vspace{.01mm}

With equation (\ref{eqn:limit2}) in mind, note that additional convex constraints may also be used to influence the limiting behavior of the robots. Moreover, by carefully selecting the number of robots allocated to each recurrent class, one can achieve a desirable limiting behavior for the ensemble of robots.

\vspace{1mm}
\section{Conclusions}\label{sec:conc}
We have proposed methods to design memoryless strategies for controlled Markov chains that guarantee maximal persistent surveillance properties under safety constraints. The uncomplicated structure of the resulting controllers makes them implementable in small robots. We have described a finitely parametrized convex program that solves this problem via entropy maximization principles, and we show that the computed control policy results in the closed loop Markov chain with the least number of recurrent classes. 

\section*{Acknowledgment}
This work was supported by NSF Grant CNS  0931878,  AFOSR Grant FA95501110182, ONR UMD-AppEl Center and the Multiscale Systems Center, one of six research centers funded under the Focus Center Research Program.

\nocite{HernandezLerma:1996ti}
\nocite{Puterman:1994vo}
\nocite{Wolfe:1962ua}
\nocite{Borkar:1990cz}
\nocite{Fox:1966ut}

\bibliographystyle{plain}

\bibliography{ref}  

\end{document}